\documentclass[pra,groupedaddress,superscriptaddress,showkeys,twocolumn,amsmath,amssymb,a4paper]{revtex4}
\usepackage{geometry}
\usepackage[dvipsnames]{xcolor}
\usepackage{graphicx}
\usepackage{verbatim}
\usepackage{float}
\usepackage{graphics}
\usepackage{mathrsfs}
\usepackage{amssymb}
\usepackage{amsmath}
\usepackage{amsthm}
\usepackage{bm}
\usepackage{dsfont}
\usepackage{hyperref}
\usepackage{braket}
\usepackage{wrapfig}
\usepackage{amsbsy}

\usepackage[T1]{fontenc}

\newtheorem{proposition}{Proposition}
\newtheorem{proposition?}{Proposition?}

\theoremstyle{definition}

\begin{document}

\title{On quantum operations of photon subtraction and photon addition}

\author{S. N. Filippov}
\affiliation{Steklov Mathematical Institute of Russian Academy of
Sciences, Gubkina St. 8, Moscow 119991, Russia}
\affiliation{Valiev Institute of Physics and Technology of Russian
Academy of Sciences, Nakhimovskii Pr. 34, Moscow 117218, Russia}

\begin{abstract}
The conventional photon subtraction and photon addition
transformations, $\varrho \rightarrow t a \varrho a^{\dag}$ and
$\varrho \rightarrow t a^{\dag} \varrho a$, are not valid quantum
operations for any constant $t>0$ since these transformations are
not trace nonincreasing. For a fixed density operator $\varrho$
there exist fair quantum operations, ${\cal N}_{-}$ and ${\cal
N}_{+}$, whose conditional output states approximate the
normalized outputs of former transformations with an arbitrary
accuracy. However, the uniform convergence for some classes of
density operators $\varrho$ has remained essentially unknown. Here
we show that, in the case of photon addition operation, the
uniform convergence takes place for the
energy-second-moment-constrained states such that ${\rm
tr}[\varrho H^2] \leq E_2 < \infty$, $H = a^{\dag}a$. In the case
of photon subtraction, the uniform convergence takes place for the
energy-second-moment-constrained states with nonvanishing energy,
i.e., the states $\varrho$ such that ${\rm tr}[\varrho H] \geq E_1
>0$ and ${\rm tr}[\varrho H^2] \leq E_2 < \infty$. We prove that
these conditions cannot be relaxed and generalize the results to
the cases of multiple photon subtraction and addition.
\end{abstract}


\maketitle

\section{Introduction}

In quantum theory, a system state is described by a density
operator~\cite{holevo-book,heinosaari-ziman},  i.e., a
positive-semidefinite operator $\varrho$ on Hilbert space ${\cal
H}$ such that its trace ${\rm tr}[\varrho] = 1$. Denote ${\cal
S}({\cal H})$ the set of density operators on ${\cal H}$.
Hereafter in this paper, we consider a separable Hilbert space
${\cal H}$ with a countable orthonormal basis
$\{\ket{n}\}_{n=0}^{\infty}$ such that $\ket{n}\bra{n}$ is a Fock
state with the fixed number $n$ of photons is a fixed mode of
electromagnetic radiation. The photon annihilation operator $a$
and the photon creation operator $a^{\dag}$ are defined through
\begin{equation*}
a = \sum_{n=1}^{\infty} \sqrt{n} \ket{n-1}\bra{n}, \qquad a^{\dag}
= \sum_{n=0}^{\infty} \sqrt{n+1} \ket{n+1}\bra{n}
\end{equation*}
 and satisfy the commutation relation $[a,a^{\dag}] = I$, the
indentity operator on ${\cal H}$. Hereafter, $\dag$ denotes the
Hermitian conjugation. The photon creation and annihilation
operators are extensively used in quantum
optics~\cite{vogel-welsch-book} because many physical operators
and characteristics are be expressed through them, for instance in
terms of the moments ${\rm tr}[\varrho (a^{\dag})^m a^n]$,
Refs.~\cite{fm-2011,fm-jpa-2012,fm-os-2012}.

Conditional transformations of quantum states in a measurement are
conventionally described by a  mapping ${\cal I}: (\Omega,{\cal
F}) \rightarrow {\cal O}$ that is also referred to as
instrument~\cite{heinosaari-ziman,davies-1976,davies-lewis-1970,ozawa-1984}.
Here, $\Omega$ is a nonempty set of classical measurement
outcomes, ${\cal F}$ is a $\sigma$-algebra on $\Omega$, ${\cal O}$
is a set of operations on ${\cal T}({\cal H})$, and ${\cal
T}({\cal H})$ is the set of trace class operators. The definition
of quantum operation naturally follows from physical requirements,
namely, a mapping ${\cal N}$ on ${\cal T}({\cal H})$ is an
operation if it is linear, completely positive, and trace
nonincreasing. The complete positivity of ${\cal N}$ means that
the mapping ${\cal N} \otimes {\rm Id}_R$ on ${\cal T}({\cal H}
\otimes {\cal H}_R)$ is positive for all finite dimensional
extensions ${\cal H}_R$. The physical meaning of complete
positivity is related with the fact that the system in interest
can be potentially entangled with an ancillary system $R$, and the
transformation of the total density operator must be positive.
Since the ancillary system is not affected by ${\cal N}$, the
total transformation is ${\cal N} \otimes {\rm Id}_R$, where ${\rm
Id}_R$ is the identity map on ${\cal T}({\cal H}_R)$.

Let $\varrho$ be an input state and ${\cal N} = {\cal I}(x)$ be a
quantum operation associated with the classical outcome $x$ of
instrument ${\cal I}$. The quantity ${\rm tr}[{\cal N}[\varrho]]
\leq 1$ is the probability to observe the outcome $x$. Suppose
${\rm tr}\big[{\cal N}[\varrho] \big] > 0$, then
\begin{equation} \label{conditional-general}
\widetilde{\varrho} = \frac{{\cal N}[\varrho]}{{\rm tr}\big[{\cal
N}[\varrho] \big]}
\end{equation}
is a conditional output density operator associated with the
outcome $x$~\cite{heinosaari-ziman,lf-2017}.

In the physics literature, the photon subtraction transformation
${\cal A}_{-}$ and the photon addition transformation  ${\cal
A}_{+}$ are defined
through~\cite{wenger-2004,zavatta-2004,kim-2008,zavatta-2009,dodonov-2009,bellini-2010,wang-2012,kumar-2013,filippov-2013,agudelo-2017,bogdanov-2017,avosopiants-2018,barnett-2018}
\begin{equation}
\label{transformations} {\cal A}_{-}[\varrho] = t a \varrho
a^{\dag}, \qquad {\cal A}_{+}[\varrho] = t a^{\dag} \varrho a,
\end{equation}
 where $t>0$ is a real number proportional to the probability of the successful transformation. The conditional output states read
\begin{equation} \label{conditional-plus-minus}
\widetilde{\varrho}_{-} = \frac{ a \varrho a^{\dag} }{{\rm tr}[ a
\varrho a^{\dag} ]}, \qquad \widetilde{\varrho}_{+} =
\frac{a^{\dag} \varrho a}{{\rm tr}[ a^{\dag} \varrho a ]}.
\end{equation}

The transformations~(\ref{transformations}) satisfy the conditions
of linearity and complete  positivity, however, they are not trace
nonicreasing. In fact, let $\varrho = \frac{1}{\zeta(s)}
\sum_{n=1}^{\infty} \frac{1}{n^s} \ket{n} \bra{n}$, where
$\zeta(s)$ is the Riemann zeta function, $s>2$. Then ${\rm
tr}[\varrho] = 1$ and ${\rm tr}\big[ {\cal A}_{-}[\varrho] \big] =
\frac{t}{\zeta(s)} \sum_{n=1}^{\infty} \frac{1}{n^{s-1}} = t
\frac{\zeta(s-1)}{\zeta(s)} \rightarrow \infty$ if $s\rightarrow
2+0$ for all $t>0$. Similarly, ${\rm tr}\big[ {\cal
A}_{+}[\varrho] \big] = \frac{t}{\zeta(s)} \sum_{n=1}^{\infty}
\left( \frac{1}{n^{s-1}} + \frac{1}{n^s} \right) = t
\frac{\zeta(s-1) + \zeta(s)}{\zeta(s)} \rightarrow \infty$ if
$s\rightarrow 2+0$ for all $t>0$. This means that the
transformations~(\ref{transformations}) are not quantum operations
and cannot be exactly implemented in any experiment.

Recently, the quantum operations have been extended to the space
of relatively bounded
operators~\cite{shirokov-2019,shirokov-2018,shirokov-arxiv-2018},
where the bound is related with the system Hamiltonian. In the
case of the one-mode electromagnetic radiation, the Hamiltonian is
essentially the photon number operator and reads
$H = a^{\dag} a. $
The goal of this paper is to show that for for some classes of
states $\varrho$ with specific restrictions on energy moments
${\rm tr}[\varrho H^k]$ there exist fair quantum operations ${\cal
N}_{\pm}$ such that conditional output
states~(\ref{conditional-general})
and~(\ref{conditional-plus-minus}) become indistinguishable in
practice. In other words, the
transformations~(\ref{transformations}) can be realized
approximately with an arbitrary precision for all states in the
class. We also discuss the processes of multiple photon
subtraction and photon addition.

\section{Approximate photon subtraction and photon addition}

The physical model of photon subtraction exploits an ideal beam
splitter,  with one input being a state $\varrho$ and the other
(auxiliary) input being a vacuum. A detection of a single photon
in the output auxiliary mode results in the following quantum
operation~\cite{kim-2008}:
\begin{equation} \label{approximate-subtraction}
{\cal N}_{-}(\gamma) [\varrho] = (e^{2\gamma} - 1) \, a e^{-\gamma
a^{\dag} a} \, \varrho \, e^{-\gamma a^{\dag} a} a^{\dag},
\end{equation}
where $\gamma > 0$ and $e^{-2\gamma}$ is the power transmittence.
From this viewpoint, this process describes an open quantum
dynamics for the system~\cite{lvof-2019,fc-2018}. If $k$ photons
are observed in the output auxiliary mode, then one gets the
operation
\begin{equation*}
{\cal N}_{-k}(\gamma) [\varrho] = \frac{1}{k!} (e^{2\gamma} - 1)^k
\, a^k e^{-\gamma a^{\dag} a} \, \varrho \, e^{-\gamma a^{\dag} a}
(a^{\dag})^k.
\end{equation*}

It is not hard to see that $\sum_{k=0}^{\infty} {\cal
N}_{-k}^{\dag}(\gamma) [I]= I$,  i.e., $\sum_{k=0}^{\infty} {\cal
N}_{-k}(\gamma)$ is trace preserving and each ${\cal
N}_{-k}(\gamma)$ is trace nonincreasing. Therefore, the
transformation ${\cal N}_{-k}^{\dag}(\gamma)$ is a fair quantum
operation.

Similarly, if the auxiliary mode is initially in the single-photon
state and no photons are observed at its output, then one obtains
the operation of approximate photon addition
\begin{equation} \label{approximate-addition}
{\cal N}_{+}(\gamma) [\varrho] = (e^{2\gamma} - 1) \, e^{-\gamma
a^{\dag} a} a^{\dag} \, \varrho \, a e^{-\gamma a^{\dag} a}.
\end{equation}
The approximate addition of $k$ photons reads
\begin{equation*}
{\cal N}_{+k}(\gamma) [\varrho] = (e^{2\gamma} - 1) \,  e^{-\gamma
a^{\dag} a} (a^{\dag})^k  \, \varrho \, a^k e^{-\gamma a^{\dag}
a}.
\end{equation*}

It is worth mentioning that other realization of  approximate
photon addition via the spontaneous parametric down conversion are
usually implemented in practice~\cite{zavatta-2004,zavatta-2009}.

Let us demonstrate that for a general state $\varrho$ the  result
of an approximate photon
subtraction~(\ref{approximate-subtraction}) can significantly
differ from the state~(\ref{conditional-plus-minus}). The
distinguishability between two quantum states $\varrho$ and
$\sigma$ reads $\frac{1}{2} \| \varrho - \sigma \|_1$ and
quantifies the optimal minimum-error
discrimination~\cite{holevo-book,heinosaari-ziman}. Here $\|X\|_1
= {\rm tr}[\sqrt{X^{\dag} X}]$.

\begin{proposition} \label{prop-1}
For any given $\gamma > 0$ there exists a state $\varrho \in {\cal
S}({\cal H})$ with finite energy ${\rm tr}[\varrho H] < \infty$
such that
\begin{equation*}
\left\| \widetilde{\varrho}_{\pm} - \frac{{\cal
N}_{\pm}(\gamma)[\varrho]}{{\rm tr} \big[ {\cal
N}_{\pm}(\gamma)[\varrho] \big]} \right\|_1 \geq \frac{1}{2} \ln
(e-1) \approx 0.27.
\end{equation*}
\end{proposition}
\begin{proof}
We restrict to the case of photon subtraction.  The case of photon
addition is treated in a similar way. Consider a one-parameter
family of states $\varrho(s) = \frac{1}{\zeta(s)}
\sum_{n=1}^{\infty} \frac{1}{n^s} \ket{n}\bra{n}$ with $s > 2$.
Then
\begin{eqnarray} \label{inequality-prop-1-complex}
&& \left\| \widetilde{\varrho}_{-}(s) - \frac{{\cal
N}_{-}(\gamma)[\varrho(s)]}{{\rm tr} \big[ {\cal
N}_{-}(\gamma)[\varrho(s)] \big]} \right\|_1 \nonumber\\
&& = \sum_{n=1}^{\infty} \frac{1}{n^{s-1}} \left\vert
\frac{e^{-2\gamma n}}{{\rm Li}_{s-1}(e^{-2\gamma})} -
\frac{1}{\zeta(s-1)} \right\vert,
\end{eqnarray}
where ${\rm Li}_{s-1}(z)$ is the polylogarithm of order $s-1$.

If $\gamma \geq \frac{1}{2} \ln \left( \frac{e}{e-1} \right)
\approx 0.23$, then we consider the contribution of the term with
$n=1$ only and get
\begin{equation} \label{inequality-prop-1}
\left\| \widetilde{\varrho}_{-}(s) - \frac{{\cal
N}_{-}(\gamma)[\varrho(s)]}{{\rm tr} \big[ {\cal
N}_{-}(\gamma)[\varrho(s)] \big]} \right\|_1 \geq
\frac{e^{-2\gamma}}{{\rm Li}_{s-1}(e^{-2\gamma})} -
\frac{1}{\zeta(s-1)},
\end{equation}
 If $s\rightarrow
+\infty$, then the right hand side of~(\ref{inequality-prop-1})
vanishes. Since $\lim_{s \rightarrow 2+0} {\rm
Li}_{s-1}(e^{-2\gamma}) = {\rm Li}_{1}(e^{-2\gamma}) = - \ln (1 -
e^{-2\gamma})$ and $\lim_{s \rightarrow 2+0} \zeta(s-1) =
+\infty$, there exists $s_1>2$ such that
\begin{eqnarray*}
&& \frac{e^{-2\gamma}}{{\rm Li}_{s_1-1}(e^{-2\gamma})} -
\frac{1}{\zeta(s_1-1)} = \frac{e^{-2\gamma}}{2\ln \left(\frac{1}{1
- e^{-2\gamma}} \right)} \nonumber\\
&&  \geq \frac{e-1}{2e} > \frac{1}{2} \ln (e-1).
\end{eqnarray*}

If $0 < \gamma < \frac{1}{2} \ln \left( \frac{e}{e-1} \right)$,
then we consider the terms in
Eq.~(\ref{inequality-prop-1-complex}) with $n \leq N =
\left\lfloor \frac{1}{2\gamma}\left(1-\frac{{\rm
Li}_{s-1}(e^{-2\gamma})}{\zeta(s-1)} \right) \right\rfloor$ as the
the expression inside the absolute value bars in
Eq.~(\ref{inequality-prop-1-complex}) is positive in this case
because $e^{-2\gamma n} \geq 1 - 2 \gamma n$. Consequently,
\begin{eqnarray*} \label{inequality-prop-1-complex-2}
&& \left\| \widetilde{\varrho}_{-}(s) - \frac{{\cal
N}_{-}(\gamma)[\varrho(s)]}{{\rm tr} \big[ {\cal N}_{-}
(\gamma)[\varrho(s)] \big]} \right\|_1 \nonumber\\
&& \geq \sum_{n=1}^{N} \frac{1}{n^{s-1}} \left( \frac{e^{-2\gamma
n}}{{\rm Li}_{s-1}(e^{-2\gamma})} - \frac{1}{\zeta(s-1)} \right) \nonumber\\
&& \rightarrow \frac{1}{{\rm Li}_1(e^{-2\gamma})}
\sum_{n=1}^{\lfloor \frac{1}{2\gamma} \rfloor} \frac{e^{-2\gamma
n}}{n}
\end{eqnarray*}
 if $s \rightarrow 2+0$. Therefore, there exists $s_2 > 2$ such that
\begin{eqnarray*} \label{inequality-prop-1-complex-3}
&& \left\| \widetilde{\varrho}_{-}(s_2) - \frac{{\cal
N}_{-}(\gamma)[\varrho(s_2)]}{{\rm tr} \big[ {\cal
N}_{-}(\gamma)[\varrho(s_2)] \big]} \right\|_1 \geq \frac{1}{2
{\rm Li}_1(e^{-2\gamma})} \sum_{n=1}^{\lfloor \frac{1}{2\gamma}
\rfloor}
\frac{e^{-2\gamma n}}{n} \nonumber\\
&& \geq \frac{1}{2} \left( 1 - \frac{1}{{\rm Li}_1(e^{-2\gamma})} \sum_{n=\lfloor \frac{1}{2\gamma} \rfloor +1}^{\infty} \frac{e^{-2\gamma n}}{n} \right) \nonumber\\
&& \geq  \frac{1}{2} \left( 1 -
\frac{2\gamma}{e(1-e^{-2\gamma}){\rm Li}_1(e^{-2\gamma})} \right)
\geq \frac{1}{2} \ln (e-1).
\end{eqnarray*}

The energy of states $\varrho(s_1)$ and $\varrho(s_2)$ is finite
because $s_1,s_2 > 2$.
\end{proof}

Proposition~\ref{prop-1} reveals that the  physically
implementable approximation of photon subtraction or addition
cannot reproduce the result of an ideal photon subtraction or
addition~(\ref{conditional-plus-minus}) for any input state.
Physically, the problem arises due to a high energy of the input.
In the next section, we show that the conditional output
state~(\ref{conditional-general}) for the approximate
operation~(\ref{approximate-subtraction}) does not converge
uniformly to the result of the ideal photon
subtraction~(\ref{conditional-plus-minus}) even in the case  of
energy-constrained states.

\section{Energy-constrained states}

Denote ${\cal S}_E({\cal H})$ the set of states  $\varrho$ such
that $0 < {\rm tr} [ \varrho H ] \leq
E$~\cite{shirokov-2018,becker-2018,winter-2017}.

\begin{proposition} \label{prop-2}
For any given $\gamma > 0$ and $E>0$ there exists a state $\varrho
\in {\cal S}_{E+1}({\cal H})$ such that
\begin{equation} \label{equation-prop-2}
\left\| \widetilde{\varrho}_{\pm} - \frac{{\cal
N}_{\pm}(\gamma)[\varrho]}{{\rm tr} \big[ {\cal
N}_{\pm}(\gamma)[\varrho] \big]} \right\|_1 \geq
\sqrt{\frac{E}{E+1}}.
\end{equation}
\end{proposition}
\begin{proof}
In the case of photon subtraction, consider a family of states
$\varrho(N) = \ket{\psi(N)}\bra{\psi(N)}$ with $\ket{\psi(N)} =
\sqrt{1-\frac{E}{N}} \ket{1} + \sqrt{\frac{E}{N}} \ket{N}$, $N
\geq \max(E,2)$. The states in the family have the energy ${\rm
tr}[\varrho H] = 1 - \frac{E}{N} + E \leq E+1$, so $\varrho(N) \in
{\cal S}_{E+1}({\cal H})$. The conditional output density operator
for the ideal photon subtraction, $\widetilde{\varrho}_{-}(N)$,
has support spanned by vectors $\ket{0}$ and $\ket{N-1}$, so it is
given by the following matrix in the corresponding 2-dimensional
subspace:
\begin{eqnarray*}
\widetilde{\varrho}_{-}(N) & = & \frac{1}{1 - \frac{E}{N} + E}
\left(
\begin{array}{cc}
    1 - \frac{E}{N} & \sqrt{E} \sqrt{1 - \frac{E}{N}} \\
    \sqrt{E} \sqrt{1 - \frac{E}{N}} & E
\end{array} \right) \nonumber\\
& \rightarrow & \frac{1}{E+1} \left( \begin{array}{cc}
    1 & \sqrt{E} \\
    \sqrt{E} & E
\end{array} \right) \quad \text{if~} N \rightarrow \infty.
\end{eqnarray*}

On the other hand, the conditional output state for the
approximate photon subtraction has support in the same subspace
and reads
\begin{eqnarray*}
&& \frac{{\cal N}_{-}(\gamma)[\varrho(N)]}{{\rm tr} \big[ {\cal
N}_{-}(\gamma)[\varrho(N)] \big]}  =  \frac{1}{\left( 1 -
\frac{E}{N} \right) e^{-2\gamma} + E e^{-2\gamma N}} \nonumber\\
&& \times \left(
\begin{array}{cc}
   \left( 1 - \frac{E}{N} \right) e^{-2\gamma}  & \sqrt{E} \sqrt{1 - \frac{E}{N}} e^{-\gamma(N+1)} \\
   \sqrt{E} \sqrt{1 - \frac{E}{N}} e^{-\gamma(N+1)}  &  E e^{-2\gamma N}
\end{array} \right) \nonumber\\
&& \rightarrow \left( \begin{array}{cc}
    1 & 0 \\
    0 & 0
\end{array} \right) \quad \text{if~} N \rightarrow \infty.
\end{eqnarray*}

Therefore, $\lim_{N \rightarrow \infty} \left\|
\widetilde{\varrho}_{-}(N) - \frac{{\cal
N}_{-}(\gamma)[\varrho(N)]}{{\rm tr} \big[ {\cal
N}_{-}(\gamma)[\varrho(N)] \big]} \right\|_1 =
2\sqrt{\frac{E}{E+1}}$ and there exists a finite $N < \infty$ such
that~(\ref{equation-prop-2}) is fulfilled.

In the case of photon addition, similarly consider a  family of
states $\varrho(N) = \ket{\psi(N)}\bra{\psi(N)}$ with
$\ket{\psi(N)} = \sqrt{1-\frac{E}{N}} \ket{0} + \sqrt{\frac{E}{N}}
\ket{N}$, $N \geq \max(E,1)$.
\end{proof}

The physical meaning of Proposition~\ref{prop-2} is that in a
fixed experimental scheme it is impossible to obtain the uniform
convergence of the approximate photon subtraction (addition) to
the ideal one within the set of energy-constrained states with
fixed $E$. In other words, there exist states with the same energy
such that for one of them the approximate photon subtraction is
very close to the ideal photon subtraction, whereas for another
one it is quite far from ideal.

Note that the mapping~(\ref{conditional-plus-minus})  transforms
the energy-constrained states in the proof of
Proposition~\ref{prop-2} to the states $\widetilde{\varrho}_{\pm}$
with unbounded energy, i.e., for any $E>0$ and $E'>0$ there exists
a state $\varrho \in {\cal S}_{E+1}({\cal H})$ such that
$\widetilde{\varrho}_{\pm} \not\in {\cal S}_{E'}({\cal H})$.

Analyzing the states in the proof of Proposition~\ref{prop-2},  we
observe that $\lim_{N \rightarrow \infty} {\rm tr}[\varrho H] =
E+1$ whereas $\lim_{N \rightarrow \infty} {\rm tr}[\varrho H^2] =
\infty$. This allows one to make a conjecture that if the second
moment of Hamiltonian, ${\rm tr}[\varrho H^2]$, would be bounded
from above, there could be a uniform convergence within the set of
such states. This is indeed the case for the photon addition;
however, this is not the case for the photon subraction as we show
in the next section.

\section{Energy-second-moment-constrained~states}

Denote ${\cal S}_E^{(2)}({\cal H})$ the set of  states $\varrho$
such that $0 < {\rm tr} [ \varrho H^2 ] \leq E$. Note that
$\varrho \in {\cal S}_E^{(2)}({\cal H})$ implies $\varrho \in
{\cal S}_E({\cal H})$ because $\sum_{n=0}^{\infty} p_n n \leq
\sum_{n=0}^{\infty} p_n n^2 \leq E$ for any probability
distribution $\{p_n\}$. The mapping~(\ref{conditional-plus-minus})
transforms the energy-second-moment-constrained states to the
energy-constrained states.

\begin{proposition} \label{prop-3}
For any $\varepsilon > 0$ and $E < \infty$ there  exists $\gamma >
0$ such that
\begin{equation*} \left\| \widetilde{\varrho}_{+} -
\frac{{\cal N}_{+}(\gamma)[\varrho]}{{\rm tr} \big[ {\cal
N}_{+}(\gamma)[\varrho] \big]} \right\|_1 < \varepsilon
\end{equation*}
\noindent for all $\varrho \in {\cal S}_{E}^{(2)}({\cal H})$.
\end{proposition}
\begin{proof}
Consider a pure state $\varrho = \ket{\psi}\bra{\psi}$,  where
$\ket{\psi} = \sum_{n=0}^{\infty} c_n \ket{n}$, $\sum_n |c_n|^2 =
1$. Note that $\widetilde{\varrho}_{+} =
\ket{\varphi}\bra{\varphi}$ with $\ket{\varphi} = \left[
\sum_{k=1}^{\infty} |c_k|^2 (k+1) \right]^{-1/2}
\sum_{n=1}^{\infty} c_n \sqrt{n+1} \ket{n+1}$ and $\frac{{\cal
N}_{+}(\gamma)[\varrho]}{{\rm tr} \big[ {\cal
N}_{+}(\gamma)[\varrho] \big]} = \ket{\chi}\bra{\chi}$ with
\[
\ket{\chi} = \frac{\sum_{n=1}^{\infty} c_n \sqrt{n+1} e^{-\gamma
(n+1)} \ket{n+1}}{\sqrt{\sum_{k=1}^{\infty} |c_k|^2 (k+1)
e^{-2\gamma (k+1)}}} .
\]
Since $\| \, \ket{\varphi}\bra{\varphi} - \ket{\chi}\bra{\chi} \,
\|_1 = 2 \sqrt{1-|\braket{\varphi | \chi}|^2}$, we have
\begin{eqnarray*}
&& \left\| \widetilde{\varrho}_{+} -  \frac{{\cal
N}_{+}(\gamma)[\varrho]}{{\rm tr} \big[ {\cal
N}_{+}(\gamma)[\varrho] \big]} \right\|_1 \nonumber\\
&& = 2 \sqrt{1 - \frac{\left[ \sum\limits_{n=1}^{\infty} |c_n|^2
(n+1) e^{-\gamma (n+1)} \right]^2}{\left[
\sum\limits_{k=1}^{\infty} |c_k|^2 (k+1) e^{-2\gamma (k+1)}
\right] \left[ \sum\limits_{k=1}^{\infty} |c_k|^2 (k+1) \right]}}.
\end{eqnarray*}

Denote $F_1 = \sum_{k=0}^{\infty} |c_k|^2 k$ the energy of the
input state and $F_2 = \sum_{k=0}^{\infty} |c_k|^2 k^2 \leq E <
\infty$ the energy second moment. Then
\begin{eqnarray*}
&& \!\!\!\!\!\!\!\!\!\! F_1 + 1 \geq \sum_{n=0}^{\infty} |c_n|^2
(n+1) e^{-\gamma
(n+1)} \nonumber\\
&& \!\!\!\!\!\!\!\!\!\! \geq \sum_{n=0}^{\infty} |c_n|^2 (n+1)
(1-\gamma - \gamma n) = F_1+1 - \gamma (1 + 2F_1 + F_2).
\end{eqnarray*}
Therefore
\begin{eqnarray} \label{equation-bound-prop-3}
&& \left\| \widetilde{\varrho}_{+} - \frac{{\cal
N}_{+}(\gamma)[\varrho]}{{\rm tr}  \big[ {\cal
N}_{+}(\gamma)[\varrho] \big]} \right\|_1 \nonumber\\
&& \leq 2 \sqrt{1 - \left( \frac{ F_1+1 - \gamma (1 + 2F_1 + F_2)
}{F_1 + 1} \right)^2 } \nonumber\\
&& < \sqrt{\frac{8 \gamma (1 + 2F_1 + F_2)}{F_1 + 1}} .
\end{eqnarray}
Note that $0 \leq F_1 \leq F_2 \leq E$ so $\left\|
\widetilde{\varrho}_{+} - \frac{{\cal N}_{+}
 (\gamma)[\varrho]}{{\rm tr} \big[ {\cal N}_{+}(\gamma)[\varrho] \big]} \right\|_1 \leq \sqrt{8 \gamma (3E+1)} < \varepsilon$ if $\gamma = \frac{\varepsilon^2}{8 (3E + 2)}$.

For a mixed state $\varrho$ with the spectral decomposition
$\varrho = \sum_{i} p_i \ket{\psi_i}\bra{\psi_i}$ we  use the
purification $\ket{\Psi} = \sum_i \sqrt{p_i} \ket{\psi_i} \otimes
\ket{\psi_i} \in {\cal H} \otimes {\cal H}$ such that $\varrho =
{\rm tr}_2 \ket{\Psi}\bra{\Psi}$, where ${\rm tr}_2$ is a channel
describing the partial trace over the second subsystem, ${\rm
tr}_2 [ \cdot ] = \sum_{n=0}^{\infty} I \otimes \bra{n} \cdot I
\otimes \ket{n}$. Denote
\begin{eqnarray*}
&& \ket{\Phi} = \frac{\sum_i \sqrt{p_i} a^{\dag} \ket{\psi_i}
\otimes \ket{\psi_i}}{\sqrt{\sum_i p_i \bra{\psi_i} a a^{\dag}
\ket{\psi_i}}}, \nonumber\\
&& \ket{X} =  \frac{\sum_i \sqrt{p_i} e^{-\gamma a^{\dag}a}
a^{\dag} \ket{\psi_i} \otimes \ket{\psi_i}}{\sqrt{\sum_i p_i
\bra{\psi_i} a e^{-2\gamma a a^{\dag}} a^{\dag} \ket{\psi_i}}},
\end{eqnarray*}
then $\widetilde{\varrho}_{+} = {\rm tr}_2 \ket{\Phi}\bra{\Phi}$
and $\frac{{\cal N}_{+}(\gamma)[\varrho]}{{\rm tr} \big[ {\cal
N}_{+}(\gamma)[\varrho] \big]} = {\rm tr}_2 \ket{X}\bra{X}$. One
can readily see that $\braket{ \Phi | X } \geq \frac{ F_1+1 -
\gamma (1 + 2F_1 + F_2) }{F_1 + 1}$, where $F_1 = {\rm tr}[\varrho
H]$ and $F_2 = {\rm tr}[\varrho H^2]$, so $\| \,
\ket{\Phi}\bra{\Phi} - \ket{X}\bra{X} \, \|_1$ is bounded from
above by the same quantity as in
Eq.~(\ref{equation-bound-prop-3}). By the contractivity
property~(\cite{nielsen-chuang}, Theorem 9.2), $\left\|
\widetilde{\varrho}_{+} - \frac{{\cal
N}_{+}(\gamma)[\varrho]}{{\rm tr} \big[ {\cal
N}_{+}(\gamma)[\varrho] \big]} \right\|_1 \leq \| \,
\ket{\Phi}\bra{\Phi} - \ket{X}\bra{X} \, \|_1 < \varepsilon$ if
$\gamma = \frac{\varepsilon^2}{8 (3E + 2)}$.
\end{proof}

The proof of Proposition~\ref{prop-3} also provides the accuracy
of the physical  implementation of the photon addition. For a
state $\varrho$ with a finite energy $F$ and the energy variance
$\sigma_F^2$ the trace distance $\frac{1}{2} \left\|
\widetilde{\varrho}_{+} - \frac{{\cal
N}_{+}(\gamma)[\varrho]}{{\rm tr} \big[ {\cal
N}_{+}(\gamma)[\varrho] \big]} \right\|_1 <
\sqrt{\frac{2\gamma}{F+1}\left[(F+1)^2 + \sigma_F^2\right]}$.

The claim of Proposition~\ref{prop-3} cannot be extended to  the
case of photon subtraction as we demonstrate below.

\begin{proposition} \label{prop-4}
For any given $\gamma > 0$ and $E>0$ there exists a state $\varrho
\in {\cal S}_{E}^{(2)}({\cal H})$ such that
\begin{equation} \label{equation-prop-4}
\left\| \widetilde{\varrho}_{-} - \frac{{\cal
N}_{-}(\gamma)[\varrho]}{{\rm tr} \big[ {\cal
N}_{-}(\gamma)[\varrho] \big]} \right\|_1 \geq 1.
\end{equation}
\end{proposition}
\begin{proof}
Consider a family of states $\varrho(N) =
\ket{\psi(N)}\bra{\psi(N)}$ with $\ket{\psi(N)} =
\sqrt{1-\frac{E}{N^2}} \ket{0} + \sqrt{\frac{E}{2N^2}} \ket{1} +
\sqrt{\frac{E}{2N^2}} \ket{N}$, $N \geq \max(\sqrt{E},1)$. The
states in this family have  the energy second moment ${\rm
tr}[\varrho H^2] = \frac{E}{2N^2}(1+N^2)\leq E$, so $\varrho(N)
\in {\cal S}_{E}^{(2)}({\cal H})$. The conditional output density
operator for the ideal photon subtraction,
$\widetilde{\varrho}_{-}(N)$, has support spanned by vectors
$\ket{0}$ and $\ket{N-1}$, so it is given by the following matrix
in the corresponding 2-dimensional subspace:
\begin{equation*}
\widetilde{\varrho}_{-}(N) = \frac{1}{N+1} \left(
\begin{array}{cc}
    1 & \sqrt{N} \\
    \sqrt{N} & N
\end{array} \right) \rightarrow \left( \begin{array}{cc}
    0 & 0 \\
    0 & 1
\end{array} \right) \quad \text{if~} N \rightarrow \infty.
\end{equation*}

On the other hand, the conditional output state for the
approximate photon subtraction has support in the same subspace
and reads
\begin{eqnarray*}
&& \!\!\!\!\!\!\!\!\!\! \frac{{\cal
N}_{-}(\gamma)[\varrho(N)]}{{\rm tr} \big[ {\cal
N}_{-}(\gamma)[\varrho(N)] \big]}  =  \frac{1}{ e^{-2\gamma} + N
e^{-2\gamma N}} \nonumber\\
&& \!\!\!\!\!\!\!\!\!\!  \times \left(
\begin{array}{cc}
   e^{-2\gamma}  & \sqrt{N} e^{-\gamma(N+1)} \\
   \sqrt{N} e^{-\gamma(N+1)}  &  N e^{-2\gamma N}
\end{array} \right) \rightarrow \left( \begin{array}{cc}
    1 & 0 \\
    0 & 0
\end{array} \right) \text{~if~} N \rightarrow \infty.
\end{eqnarray*}

Therefore, $\lim_{N \rightarrow \infty} \left\|
\widetilde{\varrho}_{-}(N) - \frac{{\cal
N}_{-}(\gamma)[\varrho(N)]}{{\rm tr}  \big[ {\cal
N}_{-}(\gamma)[\varrho(N)] \big]} \right\|_1 = 2$ and there exists
a finite $N < \infty$ such that~(\ref{equation-prop-4}) is
fulfilled.
\end{proof}

The feature of states used in the proof of
Proposition~\ref{prop-4} is that their energy  ${\rm
tr}[\varrho(N) H] \rightarrow 0$ if $N \rightarrow \infty$.
Finally, we can formulate the necessary conditions for the uniform
convergence of $\widetilde{\varrho}_{-}$ to $\frac{{\cal
N}_{-}(\gamma)[\varrho]}{{\rm tr} \big[ {\cal
N}_{-}(\gamma)[\varrho] \big]}$ within a given set of states
${\cal S}'({\cal H})$: the set ${\cal S}'({\cal H})$ should be
isolated from the states with infinite energy second moment and
isolated from the states with infinitesimal energy. We show in the
next section, that these conditions are also sufficient.

\section{Energy-second-moment-constrained states with nonvanishing energy}

Denote ${\cal S}_{E_1;E_2}^{(1;2)}({\cal H})$ the set of states
$\varrho$ such that ${\rm tr}[\varrho H] \geq E_1$ and ${\rm tr} [
\varrho H^2 ] \leq E_2$.

\begin{proposition} \label{prop-5}
For any $\varepsilon > 0$, $E_1>0$, and $E_2 < \infty$ there
exists $\gamma > 0$ such that
\begin{equation*}
\left\| \widetilde{\varrho}_{-} - \frac{{\cal
N}_{-}(\gamma)[\varrho]}{{\rm tr} \big[ {\cal
N}_{-}(\gamma)[\varrho] \big]} \right\|_1 < \varepsilon
\end{equation*}
\noindent for all $\varrho \in {\cal S}_{E_1;E_2}^{(1;2)}({\cal
H})$.
\end{proposition}
\begin{proof}
Consider a pure state $\varrho = \ket{\psi}\bra{\psi}$, where
$\ket{\psi} = \sum_{n=0}^{\infty} c_n \ket{n}$,  $\sum_n |c_n|^2 =
1$. Note that $\widetilde{\varrho}_{-} =
\ket{\varphi}\bra{\varphi}$ with $\ket{\varphi} = \left(
\sum_{k=1}^{\infty} |c_k|^2 k \right)^{-1/2} \sum_{n=1}^{\infty}
c_n \sqrt{n} \ket{n-1}$ and $\frac{{\cal
N}_{-}(\gamma)[\varrho]}{{\rm tr} \big[ {\cal
N}_{-}(\gamma)[\varrho] \big]} = \ket{\chi}\bra{\chi}$ with
$\ket{\chi} = \left( \sum_{k=1}^{\infty} |c_k|^2 k e^{-2\gamma k}
\right)^{-1/2} \sum_{n=1}^{\infty} c_n \sqrt{n} e^{-\gamma n}
\ket{n-1}$. Since $\| \, \ket{\varphi}\bra{\varphi} -
\ket{\chi}\bra{\chi} \, \|_1 = 2 \sqrt{1-|\braket{\varphi |
\chi}|^2}$, we have
\begin{eqnarray*}
&& \left\| \widetilde{\varrho}_{-} - \frac{{\cal
N}_{-}(\gamma)[\varrho]}{{\rm tr} \big[ {\cal
N}_{-}(\gamma)[\varrho] \big]} \right\|_1 \nonumber\\
&& =  2 \sqrt{1 - \frac{\left( \sum_{n=1}^{\infty} |c_n|^2 n
e^{-\gamma n} \right)^2}{\left( \sum_{k=1}^{\infty} |c_k|^2 k
e^{-2\gamma k} \right) \left( \sum_{k=1}^{\infty} |c_k|^2 k
\right)}}.
\end{eqnarray*}

Denote $F_1 = \sum_{k=0}^{\infty} |c_k|^2 k \geq E_1 > 0$ the
energy of the input state and $F_2 =  \sum_{k=0}^{\infty} |c_k|^2
k^2 \leq E_2 < \infty$ the energy second moment. Then
\begin{equation*}
F_1 \geq \sum_{n=0}^{\infty} |c_n|^2 n e^{-\gamma n} \geq
\sum_{n=0}^{\infty} |c_n|^2 n (1-\gamma n) = F_1 - \gamma F_2.
\end{equation*}
 Therefore
\begin{eqnarray} \label{equation-bound-prop-5}
&& \left\| \widetilde{\varrho}_{-} - \frac{{\cal
N}_{-}(\gamma)[\varrho]}{{\rm tr} \big[ {\cal
N}_{-}(\gamma)[\varrho] \big]} \right\|_1 \leq 2 \sqrt{1 -
\frac{\left( F_1 - \gamma F_2 \right)^2}{F_1^2}} \nonumber\\
&& < \sqrt{\frac{8 \gamma F_2}{F_1}} \leq \sqrt{\frac{8 \gamma
E_2}{E_1}} = \varepsilon \quad \text{if~} \gamma = \frac{E_1
\varepsilon^2}{8 E_2}.
\end{eqnarray}

For a mixed state $\varrho$ with the spectral decomposition
$\varrho = \sum_{i} p_i \ket{\psi_i}\bra{\psi_i}$ we  use the
purification $\ket{\Psi} = \sum_i \sqrt{p_i} \ket{\psi_i} \otimes
\ket{\psi_i} \in {\cal H} \otimes {\cal H}$. Denote
\begin{eqnarray*}
&& \ket{\Phi} = \frac{\sum_i \sqrt{p_i} a \ket{\psi_i} \otimes
\ket{\psi_i}}{\sqrt{\sum_i p_i \bra{\psi_i} a^{\dag}  a
\ket{\psi_i}}}, \nonumber\\
&& \ket{X} =  \frac{\sum_i \sqrt{p_i} a e^{-\gamma a^{\dag}a}
\ket{\psi_i} \otimes \ket{\psi_i}}{\sqrt{\sum_i p_i \bra{\psi_i}
a^{\dag} a e^{-2\gamma a a^{\dag}}  \ket{\psi_i}}},
\end{eqnarray*}
then $\widetilde{\varrho}_{+} = {\rm tr}_2 \ket{\Phi}\bra{\Phi}$
and $\frac{{\cal N}_{+}(\gamma)[\varrho]}{{\rm tr} \big[ {\cal
N}_{+}(\gamma)[\varrho] \big]} = {\rm tr}_2 \ket{X}\bra{X}$. One
can readily see that $\braket{ \Phi | X } \geq \frac{ F_1 - \gamma
F_2}{F_1}$, where $F_1 = {\rm tr}[\varrho H]$ and $F_2 = {\rm
tr}[\varrho H^2]$, so $\| \, \ket{\Phi}\bra{\Phi} - \ket{X}\bra{X}
\, \|_1$ is bounded from above by the same quantity as in
Eq.~(\ref{equation-bound-prop-5}). By the contractivity
property~(\cite{nielsen-chuang}, Theorem 9.2), $\left\|
\widetilde{\varrho}_{+} - \frac{{\cal
N}_{+}(\gamma)[\varrho]}{{\rm tr} \big[ {\cal
N}_{+}(\gamma)[\varrho] \big]} \right\|_1 \leq \| \,
\ket{\Phi}\bra{\Phi} - \ket{X}\bra{X} \, \|_1 < \varepsilon$ if
$\gamma = \frac{E_1 \varepsilon^2}{8 E_2}$.
\end{proof}

The proof of Proposition~\ref{prop-5} also provides the accuracy
of the physical implementation of the photon subtraction. For a
state $\varrho$ with a finite energy $F$ and the energy variance
$\sigma_F^2$ the trace distance $\frac{1}{2} \left\|
\widetilde{\varrho}_{-} - \frac{{\cal
N}_{-}(\gamma)[\varrho]}{{\rm tr} \big[ {\cal
N}_{-}(\gamma)[\varrho] \big]} \right\|_1 <
\sqrt{\frac{2\gamma}{F}(F^2 + \sigma_F^2)}$.

\section{Discussion and conclusions}

We have clarified that the ideal
transformations~(\ref{transformations})  cannot be realized in any
experiment because the corresponding maps are not trace
nonincreasing. However, it is experimentally feasible to implement
the operations~(\ref{approximate-subtraction})
and~(\ref{approximate-addition}) of approximate photon subtraction
and addition, respectively. However, in an experiment the
transmittence parameter $e^{-2\gamma}$ is usually fixed and the
natural question arises: What are the input states $\varrho$ such
the conditional output states of approximate operations are
$\varepsilon$-close to the ideal
states~(\ref{conditional-plus-minus})? This formulation of the
problem assumes the uniform convergence of conditional output
quantum states to the ideal states~(\ref{conditional-plus-minus}).
In this paper, we sequentially imposed restrictions on input
quantum states $\varrho$. Firstly, we showed that states $\varrho$
should have finite energy. Secondly, we demonstrated that the
finite energy second moment is also necessary. This turned out to
be sufficient for the photon addition operation, however, not
sufficient for the photon subtraction operation, for which one
more restriction is to be imposed: the input states must not have
vanishing energy. The proofs of Propositions~\ref{prop-3}
and~\ref{prop-5} provide the upper bound on the error of
approximate photon addition and subtraction, respectively.

Finally, the multiple photon addition and subtraction operations
can be treated in the same way because
\begin{equation} \label{relation-multiple-single}
{\cal N}_{-}^k(\gamma)[\varrho] \propto {\cal
N}_{-k}(k\gamma)[\varrho], \qquad {\cal N}_{+}^k(\gamma)[\varrho]
\propto {\cal N}_{+k}(k\gamma)[\varrho],
\end{equation}
where the notation $X \propto Y$ for operators $X$ and $Y$ means
$X = kY$ for some constant $k$.
Eq.~(\ref{relation-multiple-single}) implies that $\frac{{\cal
A}_{+}^k[\varrho]}{{\rm}\big[ {\cal A}_{+}^k[\varrho] \big]}$
converges uniformly to $\frac{{\cal N}_{+}^k[\varrho]}{{\rm}\big[
{\cal N}_{+}^k[\varrho] \big]}$ for
energy-$(k+1)$th-moment-constrained states $\varrho$ such that
${\rm tr}[\varrho H^{k+1}] \leq E < \infty$. Simirlary,
Eq.~(\ref{relation-multiple-single}) implies that $\frac{{\cal
A}_{-}^k[\varrho]}{{\rm}\big[ {\cal A}_{-}^k[\varrho] \big]}$
converges uniformly to $\frac{{\cal N}_{-}^k[\varrho]}{{\rm}\big[
{\cal N}_{-}^k[\varrho] \big]}$ for
energy-$(k+1)$th-moment-constrained states $\varrho$ with
nonvanishing energy such that ${\rm tr}[\varrho H] \geq E_1 > 0$
and ${\rm tr}[\varrho H^{k+1}] \leq E_2 < \infty$.

Interestingly, in contrast to the quantum
channels~\cite{fm-2018,ffk-2018,f-2018,fk-2019},  the quantum
informational properties of quantum operations such as capacities
and entanglement degradation remain essentially unstudied. From
this viewpoint, the fair quantum
operations~(\ref{approximate-subtraction})
and~(\ref{approximate-addition}) can be analyzed as paradigmatic
examples of operations on continuous-variable quantum states. In
turn, the quantum operations~(\ref{approximate-subtraction})
and~(\ref{approximate-addition}) can be replaced by simpler
transformations~(\ref{transformations}) in the domain of
second-moment-energy-constrained states with non-vanishing energy.

\section{Acknowledgements}
The author thanks Maksim Shirokov and Guillermo
Garc\'{\i}a-P\'{e}rez for  fruitful discussions. The study is
supported by the Russian Science Foundation under Project No.
19-11-00086.

\end{document}